\documentclass{article}

\usepackage[T1]{fontenc} % uses good PDF fonts
\usepackage{lmodern} % version of computer modern that supports T1 encoding
\usepackage[utf8]{inputenc} % handles UTF characters in TeX source
\usepackage{cmap} % makes PDFs consistently copy pastable

\usepackage[pdfstartview=FitH,pdfpagemode=None,colorlinks,linkcolor=blue,filecolor=blue,citecolor=red,urlcolor=red]{hyperref}
\usepackage{fullpage}
\usepackage{xspace}  % define \newcommands that don't eat spaces

\usepackage{multirow}
\usepackage{float}
\usepackage{amsmath,amsthm,amssymb,xcolor}

\usepackage[capitalize]{cleveref}

\usepackage{mdframed}
\usepackage{authblk}
\usepackage{graphicx}
\usepackage{tikz}
\newif\ifnotes
\notesfalse

\newtheorem{theorem}{Theorem}[section]

\newtheorem{lemma}[theorem]{Lemma}

\Crefname{importedtheorem}{Imported Theorem}{Imported Theorems}
\Crefname{theorem}{Theorem}{Theorems}
\Crefname{proposition}{Proposition}{Propositions}
\Crefname{claim}{Claim}{Claims}
\Crefname{lemma}{Lemma}{Lemmas}
\Crefname{conjecture}{Conjecture}{Conjectures}
\Crefname{corollary}{Corollary}{Corollaries}
\Crefname{construction}{Construction}{Constructions}
\Crefname{property}{Property}{Properties}

\theoremstyle{definition}
\newtheorem{definition}[theorem]{Definition}

\Crefname{definition}{Definition}{Definitions}
\Crefname{assumption}{Assumption}{Assumptions}
\Crefname{notation}{Notation}{Notations}

\theoremstyle{remark}

\Crefname{question}{Question}{Questions}
\Crefname{remark}{Remark}{Remarks}
\Crefname{comment}{Comment}{Comments}
\Crefname{fact}{Fact}{Facts}

\def\bbR{{\mathbb R}}

\newcommand{\R}{\bbR}

\newcommand{\PSPACE}{\mathsf{PSPACE}}
\newcommand{\FPSPACE}{\mathsf{PSPACE}}

\newcommand{\NP}{\mathsf{NP}}

\newcommand{\TFNP}{\mathsf{TFNP}}
\newcommand{\TFPSPACE}{\mathsf{TFPSPACE}}

\newcommand{\smmmove}{\textsc{SMMMove}}

\newcommand{\mpmove}{\textsc{MPMove}}
\newcommand{\dechex}{\textsc{DecisionHex}}
\newcommand{\decpos}{\textsc{DecisionPoset}}

\newif\ifdraft
\drafttrue
\ifdraft
\notestrue
\else
\notesfalse
\fi

\title{Strategy-Stealing is Non-Constructive}
\author[1]{Greg Bodwin\thanks{gregory.bodwin@cc.gatech.edu.  Supported in part by NSF awards CCF-1717349,  DMS-183932 and CCF-1909756.}}
\author[2]{Ofer Grossman\thanks{ogrossma@mit.edu. Supported by the Fannie and John Hertz Foundation fellowship, an NSF GRFP award, NSF CNS-1413920,  DARPA/NJIT  491512803, Sloan  Foundation 996698, and MIT/IBM W1771646. This work was done in part at the Simons Institute for the Theory of Computing.}}
\affil[1]{Georgia Tech}
\affil[2]{MIT}
\date{}

\begin{document}

\maketitle

\thispagestyle{empty}

\begin{abstract}
In many combinatorial games, one can prove that the first player wins under best play using a simple but non-constructive argument called \emph{strategy-stealing}.
This work is about the complexity behind these proofs: how hard is it to actually find a winning move in a game, when you know by strategy-stealing that one exists? 
We prove that this problem is $\PSPACE$-Complete already for \emph{Minimum Poset Games} and \emph{Symmetric Maker-Maker Games}, which are simple classes of games that capture two of the main types of strategy-stealing arguments in the current literature.
\end{abstract}

%\section*{stuff to change}

%\gnote{I can't think of a title that I love.  Ideas here?}

%\gnote{We should discuss the intro flow, too much to type here.  In short: I like the TFNP start and tried to write that.  I took the view that Hex and Chomp are two types of SS rather than Chomp being "true SS" and Hex being something else.}

%\onote{An idea: maybe instead of talking about symmetric maker maker games with function f, we instead talk only about maker-maker games in the classical setting where $W_1 = W_2$. This will simplify the definitions a bit. We also would have to change the intro to not talk about hex. idk}

%\onote{I think looking through the intro it's hard to know what our result is -- it's kind of hidden. maybe in the our contrubution section should be more obvious or something}

%\onote{it might be good to start by talking about tfnp --> tfpspace right off the bat. I think this is the sort of thing that's very easy to digest for anyone who's familiar with tfnp. i wrote a two paragraph thing for now (I don't love the current version)}

\clearpage

\pagenumbering{arabic}

\section{Introduction}

Theoretical Computer Science includes a rich theory of the complexity class $\TFNP$, defined as the set of $\NP$ search problems where a solution always exists.
The interesting subclasses of $\TFNP$ are based on simple yet non-constructive existence proofs for these solutions.
For example: given a circuit
$$C:\{0, 1\}^n \to \{0, 1\}^{n-1},$$
one sees immediately by the Pigeonhole Principle that there exist distinct inputs $x_1, x_2$ with matching output $C(x_1) = C(x_2)$.
But can one find such a pair of inputs computationally?
This problem is complete for a complexity class $\mathsf{PWPP} \subseteq \TFNP$, and a similar story holds for various other problems with other non-constructive proofs of solution existence.

A major motivation for $\TFNP$ as an object of study is that it gives satisfying formalizations of the natural question of whether a type of proof is constructive (``is the Pigeonhole Principle constructive?'' roughly corresponds to ``is $\mathsf{P} = \mathsf{PWPP}$?'').
But really, some non-constructive proof methods in mathematics do not correspond to $\NP$ search problems at all.
Thus, we argue, a valuable direction for research in the spirit of $\TFNP$ is to look outside $\TFNP$ itself to analyze the constructiveness of proofs in other complexity classes.
This paper is about one such instance: \emph{strategy-stealing proofs}, which are fundamental existence results in combinatorial game theory that naturally lie in $\PSPACE$.

\subsection{Combinatorial Games and Strategy-Stealing}

A \emph{combinatorial game} is a finite two-player game of perfect information.
The players take turns choosing \emph{moves} that manipulate a game board by some predefined rules, eventually reaching one out of a set of \emph{terminal states} which determine the outcome of the game.
Examples of combinatorial games include chess, go, tic-tac-toe, and some others that we will describe in detail shortly.

In general, deciding which player has a winning strategy in a combinatorial game is computationally hard.
However, certain classes of games  admit slick proofs that a certain player wins under best play (and thus determining who has a winning strategy in these games is not computationally hard).
A famous example is the game \emph{Hex}, in which two players named Red and Blue alternately color in hexagons in a symmetric board (pictured in Figure \ref{fig:hex}); Red wins if there is a continuous path of red hexagons connecting the top and bottom, and Blue wins if there is a continuous path of blue hexagons connecting the left and right.
The \emph{Hex Theorem} states that exactly one of the two players will achieve a winning configuration once all hexagons have been colored, so there are no draws.
It was observed by Nash \cite{nashhex} that:

\begin{figure}[t]
    \centering
    \includegraphics[scale=0.5]{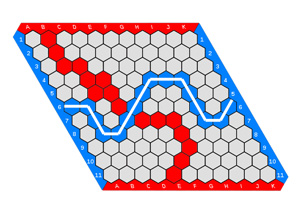}
    \caption{A Hex board with a winning configuration for Blue \cite{hexpic}.}
    \label{fig:hex}
\end{figure}

\begin{theorem} [\cite{nashhex}] \label{thm:hexwin}
The first player has a winning strategy in the game of Hex.
\end{theorem}
\begin{proof} [Proof Sketch]
Suppose for contradiction that the second player has a winning strategy.
The first player can then make an arbitrary first move and then ``steal'' the winning strategy of the second player. That is, he will now pretend he is the second player, and play according to the second player's winning strategy. This will lead to a win for the first player anyways since their arbitrary initial move can only help them achieve a winning configuration.
\end{proof}
This has been dubbed the first \emph{strategy-stealing} proof, referring to a now-broad collection of proofs that assume for contradiction that the second player can win, then repurpose the winning strategy to create a win for the first player.
Another illustrative example is the game \emph{Chomp}.
Here, the game board is an $m \times n$ chocolate bar in which the top right square has been poisoned.
The players alternately choose an uneaten square, and then eat that square and all remaining squares down and to the left.
A player loses if they eat the poisoned square.

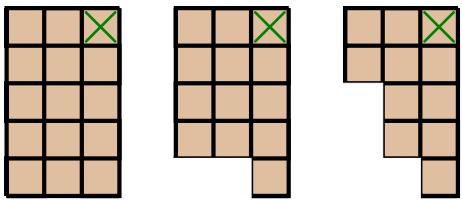
\begin{figure}[h!]
\begin{center}
\begin{tikzpicture} [scale=0.5]
\begin{scope}
\draw [fill=brown!50] (0, 0) -- (3, 0) -- (3, 5) -- (0, 5) -- cycle;
\draw [ultra thick] (3, 0) -- (0, 0) -- (1, 0) -- (1, 5) -- (2, 5) -- (2, 0) -- (3, 0) -- (3, 5);
\draw [ultra thick] (0, 5) -- (0, 0) -- (0, 1) -- (3, 1) -- (3, 2) -- (0, 2) -- (0, 3) -- (3, 3) -- (3, 4) -- (0, 4) -- (0, 5) -- (3, 5);
\node at (2.5, 4.5) {\Huge $\color{green!50!black}\mathbf{\times}$};
\end{scope}

\begin{scope}[shift={(4.5, 0)}]
\draw [fill=brown!50] (0, 0) -- (3, 0) -- (3, 5) -- (0, 5) -- cycle;
\draw [ultra thick] (3, 0) -- (0, 0) -- (1, 0) -- (1, 5) -- (2, 5) -- (2, 0) -- (3, 0) -- (3, 5);
\draw [ultra thick] (0, 5) -- (0, 0) -- (0, 1) -- (3, 1) -- (3, 2) -- (0, 2) -- (0, 3) -- (3, 3) -- (3, 4) -- (0, 4) -- (0, 5) -- (3, 5);
\node at (2.5, 4.5) {\Huge $\color{green!50!black}\mathbf{\times}$};

\draw [white, fill=white] (-1, -1) -- (2, -1) -- (2, 1) -- (-1, 1) -- cycle;
\end{scope}

\begin{scope}[shift={(9, 0)}]
\draw [fill=brown!50] (0, 0) -- (3, 0) -- (3, 5) -- (0, 5) -- cycle;
\draw [ultra thick] (3, 0) -- (0, 0) -- (1, 0) -- (1, 5) -- (2, 5) -- (2, 0) -- (3, 0) -- (3, 5);
\draw [ultra thick] (0, 5) -- (0, 0) -- (0, 1) -- (3, 1) -- (3, 2) -- (0, 2) -- (0, 3) -- (3, 3) -- (3, 4) -- (0, 4) -- (0, 5) -- (3, 5);
\node at (2.5, 4.5) {\Huge $\color{green!50!black}\mathbf{\times}$};

\draw [white, fill=white] (-1, -1) -- (2, -1) -- (2, 1) -- (1, 1) -- (1, 3) -- (-1, 3) -- cycle;
\end{scope}
\end{tikzpicture}
\end{center}
\caption{\label{fig:chomp} A valid two-move sequence from the starting position in $5 \times 3$ Chomp.}
\end{figure}

The following strategy stealing argument applies to Chomp:
\begin{theorem} [Folklore] \label{thm:chompwin}
The first player has a winning strategy in the game of Chomp.
\end{theorem}
\begin{proof} [Proof Sketch]
Consider the possible first move where the first player chomps off only the bottom-left-most square.
There are two cases.
Maybe the second player does not have a winning response, in which case the game is a win for the first player.
Alternately, suppose the second player has a winning response by chomping off an $a \times b$ block.
Since this block necessarily contains the bottom-left-most square, the board state is the same as if an $a \times b$ block had been chomped off with the first move of the game.
It follows that this $a \times b$ chomp, instead, would have been a winning first move for the first player.
\end{proof}

Both of the proofs above \textit{seem} non-constructive, in the sense that they do not yield an actual winning first move.
Our central research question is whether this is inherent:
\begin{center}
    \emph{Are strategy-stealing proofs constructive?}
\end{center}

To tackle this problem, we consider games that admit strategy stealing proofs, and we investigate the computational hardness of finding winning moves in such games.
With this view, one can see that strategy-stealing proofs can essentially be \emph{arbitrarily} non-constructive: for any combinatorial game $X$ with two players P1 and P2, we can define a game $X'$ in which the first player can decide whether he wishes to play as P1 or P2 in game $X$.\footnote{Notice the strategy stealing argument that the first player has a winning strategy in the game $X'$: suppose otherwise. Then we know if the first player chooses to play as P1, the second player has a winning strategy, so P2 has a winning strategy in $X$. But then the first player can choose to play as P2 in $X'$ and use the winning strategy for P2 in $X$.} Then finding a winning move for P1 is the same as determining the winner of $X$, which in general is computationally hard.
(We discuss this point in a little more detail in the conclusion.)

Thus, a more interesting direction is not to proceed in the maximally general case, but rather to investigate whether hardness persists in special classes of games to which strategy-stealing applies.
In particular, we will study games which (to our eye) are the minimal natural classes captured by the two strategy-stealing arguments given above.

\subsection{Our Results}

To capture ``Hex-type strategy stealing,'' we consider the well-studied class of \emph{symmetric Maker-Maker games}:
\begin{definition} [Symmetric Maker-Maker Game\footnote{This is a generalization of the usual definition: in the literature, a ``Maker-Maker game'' often implies $W_1 = W_2$.}]
In a \emph{Maker-Maker Game}, two players alternately claim elements of a finite universe $U$.
There are families of \emph{winning sets} $W_1, W_2 \subseteq P(U)$; the first player wins as soon as they claim all the elements of any winning set $S \in W_1$, the second player wins as soon as they claim all the elements of any winning set $S \in W_2$, and the game is a draw if all of $U$ is claimed without either player winning.
The game is \emph{symmetric} if $W_1, W_2$ are isomorphic, i.e., there is a permutation $\pi$ of $U$ and a bijection $\phi : W_1 \to W_2$ such that for all $S_1 = \{s_1, \dots, s_k\}\in W_1$, we have $\phi(S_1) = \{\pi(s_1), \dots, \pi(s_k)\}$. 
\end{definition}

Hex is a symmetric Maker-Maker (SMM) game, and indeed the proof of Theorem \ref{thm:hexwin} generalizes immediately to imply that any SMM game is not a win for the second player. 
There are many other examples of SMM games, which will be surveyed later.
In general an SMM game can be a draw under best play, but some games like Hex are \emph{draw-free} and so a first-player win is the only remaining possibility.
We associate a computational problem to these games as follows:

\begin{definition} [$\smmmove$] \label{def:smmmove}
The problem $\smmmove$ is defined as follows:
\begin{itemize}
    \item Input: Circuits $C_1, C_2$, both with input wires labelled $x_1, \dots, x_n$. $C_1$ and $C_2$ are the same up to relabelling of the wires. Call $X = \{x_1, \dots, x_n\}$.  The Maker-Maker game associated with this input is where $W_1$ contains any set of inputs $Y \subseteq X$ for which $C_1$ evaluates to true when the inputs in $Y$ are set to true and $X \setminus Y$ to false. $W_2$ is defined similarly with respect to $C_2$. 
    \item Output: any optimal first move for the first player in the associated game.
\end{itemize}
\end{definition}

%\onote{definition above is a promise problem}

We prove:
\begin{theorem} \label{thm:smhard}
$\smmmove$ is $\PSPACE$-hard,\footnote{A straightforward algorithm solves $\smmmove$ in polynomial space, so in some sense it is complete (ignoring subtleties in the terminology), but we will not discuss these easy upper bounds in this paper.  A similar comment holds for $\mpmove$ below.} even under the additional promise that the input defines a draw-free game with $W_1 = W_2$.
\end{theorem}

Thus, Hex-type strategy stealing is a formally non-constructive style of proof, and additional draw-freeness results like the Hex Theorem do not generally help.
To capture ``Chomp-type strategy stealing,'' we consider:

\begin{definition} [Minimum Poset Games] \label{def:posetgame}
In a \emph{poset game}, two players alternately choose remaining elements of a poset $P$, removing the chosen element and all lesser elements at each step.
A player loses if it is their turn but the poset is empty.
The game is \emph{minimum} if $P$ has a minimum element (i.e., $m \in P$ that is comparable to and less than every other element in $P$).
\end{definition}

(For both these types of games, we refer to \cite{survey1, survey2} for some of their history and prior work.)
Chomp is a minimum poset game, where the associated poset holds the squares of the chocolate bar, with the poisoned square removed, and squares are compared by the usual poset relation on $\mathbb{Z}^2$ (note that the bottom-left-most square is a minimum element).
Theorem \ref{thm:chompwin} generalizes to show that any minimum poset game is a win for the first player.
Other examples of poset games, which may or may not have a minimum, include Nim, Hackendot, certain cases of Hackenbush, and many others.
Computationally, we have:

\begin{definition} [$\mpmove$]
The problem $\mpmove$ is defined as follows:
\begin{itemize}
    \item Input: a poset $P$ (with elements and relations between them enumerated explicitly) with a minimum element.
    \item Output: any winning move for the first player in the poset game defined by $P$.
\end{itemize}
\end{definition}

\begin{theorem} \label{thm:mphard}
$\mpmove$ is $\PSPACE$-hard.
\end{theorem}

From a technical standpoint, both Theorems \ref{thm:smhard} and \ref{thm:mphard} are proved roughly as follows.
We start with a theorem in prior work stating that it is $\PSPACE$-hard to decide the winner in a related class of games: Hex from an arbitrary starting position \cite{hexpspace}, or a certain class of poset games \cite{posetpspace}.
We then apply transformations that introduce the necessary strategy-stealing properties to these games while arguing that the winner in the original game is implicitly encoded by the first player's winning move(s).
In the case of minimum poset games, this is an easy extension of the theorem in \cite{posetpspace}; we include this mostly to illustrate our conceptual goal of computationally formalizing non-constructiveness.
For SMM games, the transformation is nontrivial and requires significant new ideas.
Thus, the SMM result constitutes our main technical contribution.
%, while the result on poset games is a less technical result that still conveys some of the conceptual ideas.

\subsection{Related Work}

As mentioned, this paper is conceptually related to the study of the complexity class $\TFNP$, defined in \cite{tfnpdef} and including notable subclasses $\mathsf{PPAD, PPA, CLS, PPP, PWPP, PLS}$, among others.
These classes all hold search problems in $\NP$ that admit proofs that a solution always exists.
Our work is related in that our goal is to prove hardness of searching for a winning moves in games, when there are strategy-stealing proofs that one always exists.
The key difference is that our problems are not in NP; there is not generally a short certificate that an optimal move is indeed the first one in some optimal strategy.

In \cite{makermakerhard}, the author proves that it is $\PSPACE$-hard to decide whether an SMM game is a win for the first player or a draw.
Our work differs in that (1) to minimally generalize Hex we restrict attention to \emph{draw-free} games, in which this decision problem is trivial, and (2) we are interested in constructively finding a winning move rather than deciding existence.
Similarly related is \cite{hexpspace}, in which it is proved to be $\PSPACE$-hard to decide whether Hex from a partially-filled board is a win for the first or second player.
This can be viewed as a Maker-Maker game, but since the board is partially filled, it is not generally a symmetric Maker-Maker game and thus strategy-stealing does not apply.

There is a rich and developed theory of Maker-Maker games, poset games, and variants, most of which focuses on understanding these games under best play (rather than computational aspects of playing the games).
See books \cite{survey1, survey2} for more information.

\section{Non-Constructiveness of Strategy-Stealing}\label{mainbody}

We will now prove our main results.

\subsection{Symmetric Maker-Maker Games}

%\begin{definition} [Maker-Maker Games]
%An \emph{Asymmetric Maker-Maker Game} is a two-player game described by a finite universe enumerated $[n]$ and sets $W_1, W_2 \subseteq P([n])$ of \emph{winning sets}.
%The players alternately claim elements of $[n]$; the first player wins as soon as they claim all elements in a set in $W_1$, the second player wins as soon as they claim all elements in a set in $W_2$, or the game is a draw if all elements in $[n]$ are claimed without either player winning.
%The game is \emph{symmetric} if $W_1,  W_2$ are isomorphic (meaning that they are the same up to relabelling of the universe elements).
%\end{definition}

%\begin{definition}[Symmetric Maker Maker]
%A Maker-Maker game is Symmetric if there exists a permutation $f$ of the elements such every $S \in W_1$, satisfies $\{f(x)|x \in S\} \in W_2$.
%\end{definition}

%\begin{definition}
%We say a game is \textit{draw-free} if it is impossible for the game to end in a draw.
%\end{definition}

Our first topic will be Symmetric Maker-Maker games, and eventually a proof of Theorem \ref{thm:smhard}.

\paragraph{Examples.}

We first survey some famous examples of SMM games in the literature.

\begin{itemize}
\item Hex is an SMM game, as discussed above, which is draw-free and has non-equal winning sets.

\item In the \emph{$(n, k)$-Clique game}, the game board is a complete graph on $n$ nodes and the players take turns claiming its edges.
The first player to claim all edges in a $k$-clique wins.
This is a symmetric Maker-Maker game, even with identical winning sets $W_1 = W_2$ (i.e., the underlying permutation is the identity).
An interesting property of this game is that, for all $k$, if $n$ is sufficiently large then the game is draw-free.
This follows from \emph{Ramsey's Theorem}, which states that any $2$-coloring of the edges of the complete graph has a monochromatic clique of size $\Omega(\log n)$.
Hence, for large enough $n$, strategy-stealing implies specifically that the first player has a winning strategy.
For work on the Clique game and some natural variants, see e.g., \cite{vdwgame, gebauer2012clique, erdos1973combinatorial, beck2002positional, kusch2017problems}.

\item Tic-Tac-Toe is an SMM game, where the winning sets are the 8 possible ``lines'' in the $3 \times 3$ grid.
This game is not draw-free.

\item In $(k, d)$ Tic-Tac-Toe, the game board is the elements of the $k^d$ hypercube ($\{1, 2, \ldots, k\}^d$), and the winning sets $W_1 = W_2$ are the $k$-element subsets which are colinear in the hypercube.
The \emph{Hales-Jewett Theorem} \cite{tictactoe} implies that for every $k$, if $d$ is sufficiently large then the game is draw-free.
For work on this game, see e.g., \cite{survey1, Golomb2002}.

\item In the \emph{$(n, k)$ Arithmetic Progression game}, the universe is the set of integers $\{1, \dots, n\}$, and the winning sets are any $k$ elements that form an arithmetic progression (i.e., the difference between successive integers is equal).
\emph{Van-der-Waerden's Theorem} \cite{vdw} implies that for every $k$, if $n$ is sufficiently large then the game is draw-free.
The Arithmetic Progression game has been studied e.g., in \cite{kusch2017random, vdwgame, kusch2017problems}.
\end{itemize}

All of these games are SMM and hence admit strategy-stealing proofs that the second player does not win under best play.
The problem of finding an optimal first move for the first player can thus be captured as a special case of $\smmmove$.
All of these games are draw-free in the appropriate range of parameters (except standard Tic-Tac-Toe), and thus here they even fit the promise that the input to $\smmmove$ defines a draw-free game.
Of course, these special cases need not be as hard as $\smmmove$: for example, it is trivial to find a winning first move in the $(n, k)$-Clique game, since by symmetry of the game board all first moves are equivalent.
(Perhaps a more interesting version of the Clique game computational problem is to determine the first player's optimal move on their second turn, since the game necessarily still retains the symmetry needed for a strategy-stealing argument after each player claims only one edge.)

\paragraph{Hardness for $\smmmove$.}

We now show computational hardness for $\smmmove$. We first outline the proof ideas, and then provide a full proof.
Our starting point is the following result from prior work:
\begin{definition}[$\dechex$]
The problem $\dechex$ is defined as follows:
\begin{itemize}
    \item Input: a partially-filled Hex board $Q$
    \item Output: does Red have a winning strategy in the Hex game starting from $Q$ (assuming it is currently Red's turn to play)?
\end{itemize}
\end{definition}

\begin{theorem}[\cite{hexpspace}] \label{thm:hexpspace}
$\dechex$ is $\PSPACE$-complete.
\end{theorem}

Our goal is to reduce $\dechex$ to $\smmmove$.
First, let us remark on why Theorem \ref{thm:hexpspace} does not \emph{directly} give hardness for $\smmmove$, given that Hex is an SMM game.
The result that $\dechex$ is hard means that there exist families of positions from which deciding the winner is hard.
However, these positions are not generally symmetric, so the game starting from these positions is not SMM.
Additionally, a talented player \emph{playing from the starting position} could still potentially be able to win the game while avoiding these hard settings of the game board, thus winning without ever really encountering a $\PSPACE$-complete problem.

%\onote{this block of text below should probably be split in some way, maybe by the diagram}
So, we are given a partially filled (possibly asymmetric) board $Q$ on input to $\dechex$, representing a Hex game between Red and Blue where it is Red's turn to move, and our goal is to create a new SMM draw-free game between players First and Second that captures the structure of $Q$ in some useful way.
To build intuition, let us start with a first (incorrect) attempt at such a game $G$.
Suppose we add two new elements to the universe called $r$ and $b$.
The idea will be that claiming $r$ is choosing to play as Red in the Hex game defined by $Q$, and claiming $b$ is choosing to playing as Blue.
More formally, the winning sets $W_1 = W_2$ of the new SMM game would be:
\begin{itemize}
    \item $r$ and any set of hexagons that complete a win for Red in $Q$,
    \item $b$ and any set of hexagons that complete a win for Blue in $Q$, and
    \item $\{r, b\}$.
\end{itemize}
If $\smmmove(G) = r$, this solves $Q$: the second player must claim $b$ with their next move to block the winning set $\{r, b\}$, and then the game reduces to $Q$ itself where First plays as Red and Second plays as Blue.
Thus, if $r$ is a winning move for First, then $Q$ is a win for Red.
Unfortunately, the other cases of the proof break down.
For example: suppose the position on the board $Q$ is such that whoever has the next move wins (so $Q$ is a win for Red).
Then it is winning for First to claim \emph{either} $r$ or $b$ with their first move, meaning the output of $\smmmove(G)$ is not very informative.
Our fix is, intuitively, to amplify the game to avoid the possibility that the game winner depends on the turn order.

\begin{proof}[Proof of Theorem \ref{thm:smhard}]
We will prove that $\smmmove$ is hard by reducing $\dechex$ to it. Let $Q$ an instance of $\dechex$. We will construct a Symmetric Maker-Maker game $G$ which is draw-free and $W_1 = W_2$, such that finding a winning move in $G$ allows us to find who has a winning strategy in $Q$.

The universe of our new game $G$ will contain \emph{two} identical copies $Q_1, Q_2$ of the input to $\dechex$, as well as new elements $r, b$ like before.
The winning sets $W_1 = W_2$ in $G$ are (see Figure \ref{fig:winningsets}):
\begin{itemize}
    \item $r$ and any set of hexagons that completes a win for Red in \textit{either} $Q_1$ or $Q_2$,
    \item $b$ and any set of hexagons that completes a win for Blue in \textit{both} $Q_1$ and $Q_2$, and
    \item $\{r,b\}$.
\end{itemize}

\begin{figure}[h]
    \centering
    \begin{tikzpicture}
    \node at (0, 0) {\includegraphics[scale=0.3]{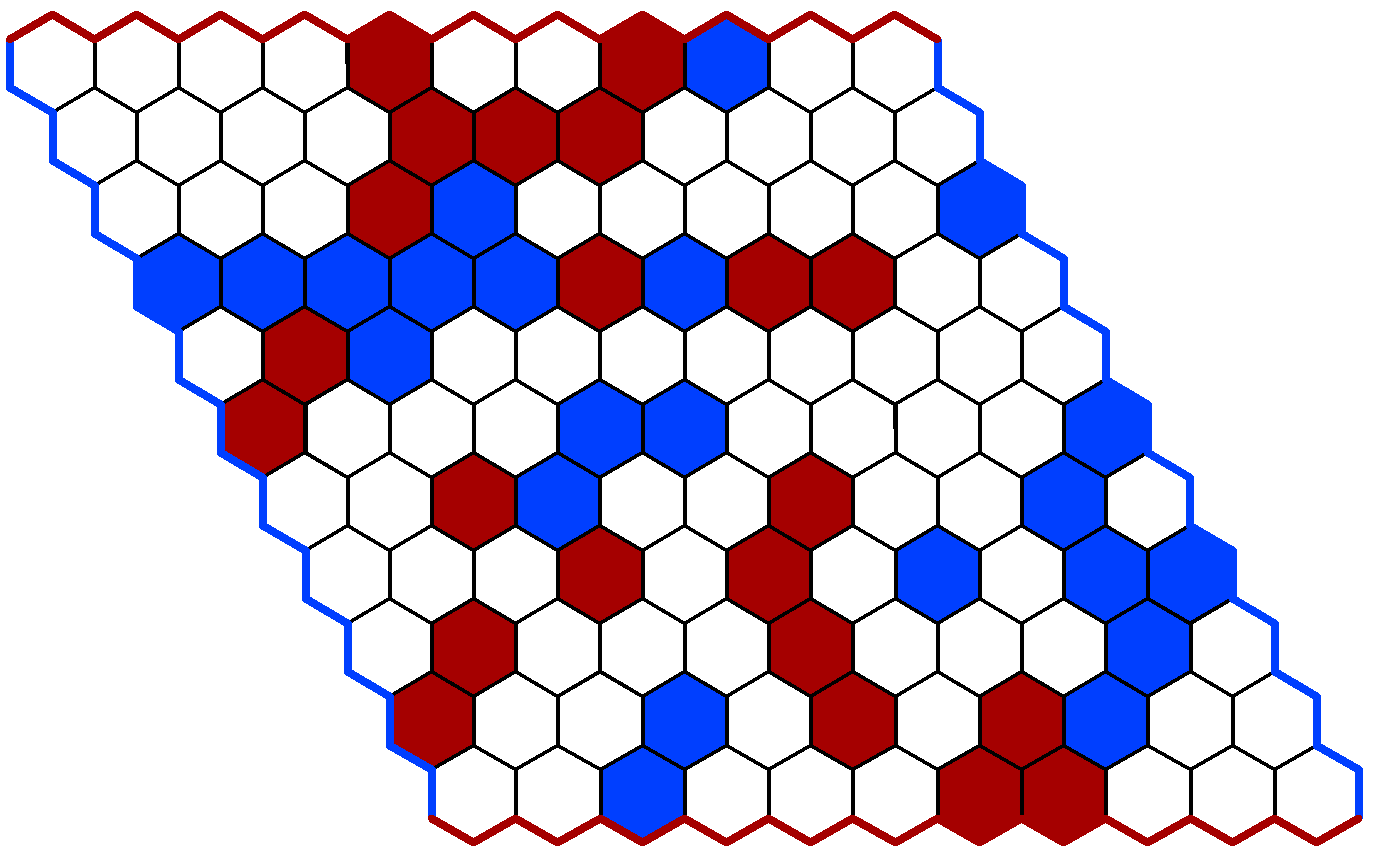}};
    \node at (8, 0) {\includegraphics[scale=0.3]{hexblank.png}};
    \node at (1, -2.5) {\Huge $Q_1$};
    \node at (9, -2.5) {\Huge $Q_2$};
    
    \node [red, align=center] at (1, 4) {Become Red\\~\\{\Huge \bf $r$}\\~\\(win on either board)};
    
    \node [blue, align=center] at (6, 4) {Become Blue\\~\\{\Huge \bf $b$}\\~\\(win on both boards)};
    
    \begin{scope}[shift={(1, 4.1)}]
    \newdimen\R
   \R=18pt
    \draw [red, ultra thick] (0:\R) \foreach \x in {60,120,...,360} {  -- (\x:\R) };
    \end{scope}
    
        \begin{scope}[shift={(6, 4.1)}]
    \newdimen\R
   \R=18pt
    \draw [blue, ultra thick] (0:\R) \foreach \x in {60,120,...,360} {  -- (\x:\R) };
    \end{scope}
    \end{tikzpicture}
    \caption{The universe used in our definition of an SMM game $G$.}
    \label{fig:winningsets}
\end{figure}

It is immediate that $G$ is an SMM game, since the winning sets are identical.
Additionally, we have:
\begin{lemma}
$G$ is draw-free.
\end{lemma}
\begin{proof}
Let $S$ be any subset of the universe in $G$.
We will show that either $S$ or its complement $S^C$ contains a winning set.
First, if $r, b \in S$ then $\{r, b\} \subseteq S$, or if $r, b \notin S$ then $\{r, b \} \subseteq S^C$.
So the nontrivial case is when $S$ contains exactly one of $r, b$; let us assume without loss of generality that $r \in S, b \notin S$ (else switch the roles of $S$ and $S^C$).
For either board $Q_i$, by the Hex Theorem and the fact that the union of $S_i, S^C_i$ covers the board, exactly one of the following two statements hold:
\begin{enumerate}
\item The elements of $S$ on the board $Q_i$ (call this $S_i$), combined with the elements on $Q_i$ which are initially marked red, form a winning configuration for red on $Q_i$.
\item
The elements of $S^C$ on the board $Q_i$ (call this set $S^C_i$), combined with the elements on $Q_i$ which are initially marked blue, form a winning configuration for blue on $Q_i$.
\end{enumerate}
Therefore, we conclude that in each board $Q_i$, either the elements of $S$ complete a win for Red (we call $Q_i$ a ``red'' board in this case), or the elements of $S^C$ complete a win for Blue (we call $Q_i$ a ``blue'' board in this case).
If at least one of $Q_1, Q_2$ is red, then $S$ contains $r$ and a Red winning set.
If both $Q_1, Q_2$ are blue, then $S^C$ contains $b$ and a Blue winning set on each board.
In either case the lemma holds.
\end{proof}

Our goal is now to show that the winning move(s) for First in $G$ completely determine the winner of $Q$.
We consider two cases:
\begin{lemma} \label{lem:bluewincase}
If Blue has a winning strategy in $Q$, then the unique winning move for First in $G$ is to claim $b$.
\end{lemma}
\begin{proof}
We first show that claiming $b$ is a winning move for First.
In response, Second is forced to claim $r$ to block the winning set $\{r, b\}$.
First then claims an arbitrary hexagon, and then each time Second claims a hexagon on $Q_1$ or $Q_2$, First claims a hexagon on the same board to execute a winning strategy for Blue.
Thus, First will have $b$ and also a winning set for Blue on both boards, meaning First wins in $G$.
(Note that Second will be unable to ever obtain a red winning set on either board, since it is not possible for both sides to obtain a winning configuration on any individual Hex board.) 

We then show that, if First does not claim $b$ with their first move, then it is a winning response for Second to claim $b$.
Here we consider two cases.
If First claims $r$, then after Second claims $b$, in each subsequent turn, each time First claims a hexagon on $Q_1$ or $Q_2$, Second can claim a hexagon on the same board to execute a winning strategy for Blue on that board, thus obtaining a winning set for Blue on both boards and hence winning in $G$.
In the other case, if First claims a hexagon in (say) $Q_1$ with their first move, then after Second claims $b$, First must immediately claim $r$ to block the winning set $\{r, b\}$.
Second then claims a hexagon on $Q_1$ and from here this case reduces to the first one.
\end{proof}

\begin{lemma} \label{lem:redwincase}
If Red has a winning strategy in $Q$, then it is not a winning move for First in $G$ to claim $b$.
\end{lemma}
\begin{proof}
Suppose that First claims $b$.
The winning response for Second is to claim $r$.
Without loss of generality, First then claims a hexagon in $Q_1$.
Second then decides to permanently ignore $Q_1$ and focus entirely on $Q_2$, claiming exclusively hexagons in $Q_2$ for the rest of the game.
Since Second is the first to move on $Q_2$, they can execute a winning strategy for Red on $Q_2$.
Thus Second will eventually claim $r$ and a winning set for Red on $Q_2$, meaning that Second wins in $G$.
(Note, again, that First cannot possibly obtain a winning set in the meantime, since they cannot possibly hold a winning set for Blue on $Q_2$.)
\end{proof}

%This leads to a new issue: if we are in the situation that whoever has the next move wins, now there is a winning strategy for the first player that begins by claiming $r$, and a winning strategy that begins by claiming $b$. To fix this, suppose that instead of a single board, we will have two boards $Q_1$ and $Q_2$. Suppose that to win you either need $r$ and to complete a win for Red on at least one board, or $b$ and to complete a win for Blue on both boards. Now, choosing $b$ is not a winning strategy for the first player (in the case that whichever of Red or Blue plays first on the board $Q$ wins)! This is because if the first player picks $b$ and the second player picks $r$, and then the first player plays on one of the boards on the next turn, the second player can guarantee to complete a win for Red on the remaining board.

%\begin{proof}
%We will show that \textsc{DecisionHex} is in $P^{\smm}$. Since \textsc{DecisionHex} is $\PSPACE$ complete, it will follow that $\PSPACE \subseteq \P^{\smm}$. Notice that $\P^{\smm} \subseteq \PSPACE$ since every game in $\smm$ can be solved in $\PSPACE$, and also $P^{\PSPACE} = \PSPACE$.

We now put the pieces together: after constructing the game $G$ as described above, from Lemmas \ref{lem:bluewincase} and \ref{lem:redwincase} we have
$$\smmmove(G) = b \qquad \text{if and only if} \qquad \lnot \dechex(Q).$$
%\onote{SMMMOVE(G) is like a set of elements techincally because it's a search problem...}
Since $\dechex$ is $\PSPACE$-complete, it follows that $\smmmove$ is $\PSPACE$-hard.
\end{proof}

\subsection{Minimum Poset Games}

Next, we prove Theorem \ref{thm:mphard}.
Our starting point is:

\begin{definition} [$\decpos$]
The problem $\decpos$ is defined as follows:
\begin{itemize}
    \item Input: a poset $P$, described by explicitly listing its elements and the relations between them.
    \item Output: is the poset game (see Definition \ref{def:posetgame}) associated to $P$ a win for the first player under best play?
\end{itemize}
\end{definition}

\begin{theorem} [\cite{posetpspace}] \label{thm:posetpspace}
$\decpos$ is $\PSPACE$-complete.
\end{theorem}

We then argue:
\begin{proof} [Proof of Theorem \ref{thm:mphard}]
Given a poset game defined by $P$, generate a new poset $P'$ by adding a new element $m$, defined to be less than every other element in $P$.
We now argue that $\mpmove(P') = m$ if and only if the original poset game defined by $P$ was a win for the second player:
\begin{itemize}
    \item Suppose $P$ is a win for the first player.  If in $P'$ the first player claims $m$ with their first move, then the game becomes equivalent to $P$ with the turn order reversed.
    Thus claiming $m$ is a losing move for the first player, and so $\mpmove(P') \ne m$.
    
    \item Suppose $P$ is a win for the second player.  If in $P'$ the first player claims $m$ with their first move, then again the game is equivalent to $P$ with the turn order reversed, so the first player has a winning strategy.
    This means we \emph{can} have $\mpmove(P') = m$, but since $\mpmove$ might return any winning move, we also need to rule out the possibility that any other move is winning.
    For this, we observe that any other first move in $P'$ necessarily removes $m$ and at least one other element from $P'$, thus giving a position that can possibly be obtained after one move in $P$.
    Since $P$ is a win for the second player, any such position must be losing for the player who creates it, and thus we have $\mpmove(P') \notin P' \setminus \{m\}$, so $\mpmove(P') = m$.
\end{itemize}
This completes the reduction from $\decpos$ to $\mpmove$, and thus $\mpmove$ is $\PSPACE$-hard.
\end{proof}

\section{Open Questions}

We conclude by listing some conceptual open questions left by this work.

\paragraph{TFPSPACE.}

Can we similarly analyze the computational properties of other interesting non-constructive proof techniques that lie outside of $\NP$?  Is there a satisfying theory of $\TFPSPACE$, in analogy with $\TFNP$?
    
\paragraph{Bounded Computational Power.}

The existence of a winning strategy in a game does not necessarily shed much light on how \emph{computationally bounded} players would play the game.
To illustrate, consider the following game: player 1 declares a circuit $C$ of their choice.
Then, player 2 wins if they can declare an input $x$ such that $C(x) = 1$.
Then, player 1 then wins if \emph{they} can declare an input $x$ with $C(x)=1$.
If both players fail to declare such an input $x$, then player 2 wins.
Here, there is clearly a winning strategy for the second player: if there exists an $x$ such that $C(x) = 1$, then declare that $x$ and win immediately; if there is no such $x$ then player 2 also wins.
However, if the players are represented by Turing machines that can only run for a polynomial amount of time, then the game is (probably) a win for player 1: for example, player 1 can pick a one way function $f$, compute it on some random $x'$ of his choice to get output $y$, and then have the circuit $C$ output $1$ on all $x$ such that $f(x) = y$.
Under standard cryptography assumptions, player 2 will be unable to find such an $x$, and then player 1 will win in the next turn by declaring the $x$ used to create the circuit.

Interestingly, for draw-free SMM games, such situations will not arise: even for computationally bounded players, it is preferable to play first, since playing an extra move is never disadvantageous.
In contrast, this is not clearly true for poset games, where the wrong first move can possibly throw the game.
Hence, this might be an interesting avenue to separate the computational properties of these two strategy-stealing arguments (since they are both $\FPSPACE$-hard under ``best play,'' i.e., unbounded computational power).
More generally, it would be interesting to further understand and formalize the effects of bounded computational power on various existence proofs for winning strategies in combinatorial game theory.
    
\paragraph{Generalized Strategy-Stealing.}

Many strategy-stealing arguments can be viewed as a reduction of the game tree to a tautology.
To illustrate, the game tree of Chomp may be phrased as follows.
Let $X$ be the subgame tree from the starting chocolate bar with the bottom-left-most square removed.
The proof of Theorem \ref{thm:chompwin} essentially observes that the Chomp game tree is equal to $X \text{ or } \lnot X$, which is true as a formula (meaning a win for the first player) regardless of the value of $X$ (see Figure \ref{fig:chompreduction}).
 
\begin{figure}[H]
    \centering
    \begin{tikzpicture}
    \draw [fill=black] (6, 0) circle [radius=0.15];
    \node at (6, 0.5) {\textbf{Starting Position}};
    \draw [ultra thick] (6, 0) -- (4, -1);
    \draw [fill=black] (4, -1) circle [radius=0.15];
    \node [align=center]at (3.5, -0.3) {(chomp bot-left\\square only)};
    \draw [ultra thick, dashed] (4, -1) -- (3, -3) -- (5, -3) -- cycle;
    \node at (4, -2.3) {\Huge $\mathbf{x}$};
    \draw [thick] (6, -0.5) -- (9, -0.5);
    \node [fill=white, align=center] at (9, -0.5) {(all other possible\\first moves)};
    
    \draw [ultra thick, dashed] (6, 0) -- (5, -2) -- (7, -2) -- cycle;
    \node at (6, -1.3) {\Huge $\mathbf{x}$};
    
    \end{tikzpicture}
    \caption{Due to the symmetry in the game tree of Chomp illustrated here, the value of the game tree can be expressed as $X \text{ or } \lnot X$, where $X$ is the game tree after a $1 \times 1$ square has been chomped.}
    \label{fig:chompreduction}
\end{figure}
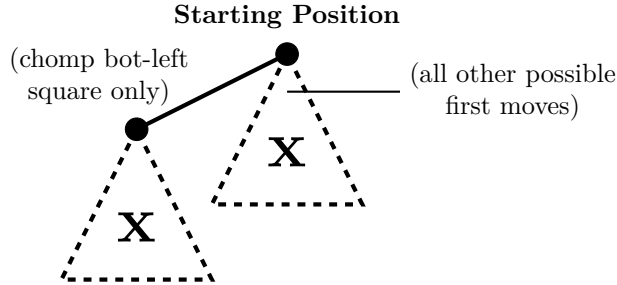

\noindent Naturally, a reduction of the game tree to \emph{any} tautology implies a first-player win, including more complicated tautologies in which multiple variables are assigned to multiple subgames.
We might call this type of argument \emph{generalized strategy-stealing}, as it extends the usual proofs in the literature that use only $X \text{ or } \lnot X$.
For any given tautology it is easy enough to invent an artificial game that admits a generalized strategy-stealing proof via that tautology.
However, it would be interesting to find a ``natural'' game that admits a generalized strategy-stealing proof, using a tautology formally distinct from $X \text{ or } \lnot X$.

\paragraph{Hardness for Specific Games.} While we have proved hardness for finding winning moves in game classes that include Hex and Chomp, our results do not imply hardness for Hex and Chomp specifically.
In particular, it would be interesting to determine whether or not the following problem is in $\mathsf{FP}$: given game board dimensions for Hex or Chomp (or basically any other game mentioned in this paper), written in unary, output a winning move for the first player.
This problem is in $\mathsf{\FPSPACE}$, but it will not be readily possible to prove it $\mathsf{\FPSPACE}$-hard for the following reason: it is known that no unary language can be $\NP$-complete unless $\mathsf{P} = \NP$; thus, a unary language complete for $\PSPACE$ would imply that $\mathsf{P} = \NP$ or $\NP \ne \PSPACE$, which is not known and would constitute a breakthrough in complexity theory.
Thus it is unclear what hardness notion should be used to approach this question.

\paragraph{Other Notions of Constructiveness.} We have proved that it is generally hard to find a winning move in a game, even when strategy-stealing arguments apply.
Finding a good first move is one natural formalization of ``constructiveness'' in $\PSPACE$, but there are others.
For example, here is an open question that we have not addressed: for (say) the game Hex, does there necessarily exist a polynomial-size circuit that plays the game optimally, even if it is computationally hard to find the circuit?

%\onote{no, because then pspace in p poly}
%\onote{switch to hex/chomp rather than general smm games}
%\onote{add more things}

%\section{Ramblings, questions, etc}

%\begin{itemize}
    %\item Can we classify ``symmetric'' set systems $\mathcal{F}$ (without referring to $f$)?
    %\item
    %boolean thing: try more tautologies to see if things other than x or not x give good stuff or if it just ends up being the same as strategy stealing
    %\item    examples of games
    %\begin{itemize}
    %    \item chomp
    %    \item any maker maker game with no draws...
    %\end{itemize}
    %\item do we want to change the MS deifinition to "draw-free maker maker games?"
    
    %\item Bounded power machine distinction
    
    %\item    artificial classes of games that are not pspace hard...
%\end{itemize}

\bibliographystyle{alpha}
\bibliography{bibfile}

\end{document}